\documentclass{amsart}

\usepackage{amsrefs}
\usepackage[utf8]{inputenc}

\newtheorem{theorem}{Theorem}[section]
\newtheorem{lemma}[theorem]{Lemma}

\theoremstyle{definition}
\newtheorem{corollary}[theorem]{Corollary}

\theoremstyle{remark}

\numberwithin{equation}{section}

\begin{document}

\title{Positivity of the two-dimensional Brown-Ravenhall operator}

\author[S. Walter]{Stefan Walter}
\address{Ludwig-Maximilians-Universität München,
Mathematisches Institut,
Theresienstr. 39, 
80333 München}
\curraddr{}
\email{stwalter@vr-web.de}
\thanks{}

\subjclass[2000]{Primary 81Q10; Secondary 47F05}

\date{}

\dedicatory{}

\begin{abstract}
We determine the critical coupling of the two-dimensional Brown-Ravenhall operator with Coulomb potential. Boundedness from below has essentially been proven by Bouzouina, whose work however contains a trivial error leading to a wrong constant exactly one half of the actual critical constant. Furthermore we show that the operator is in fact positive. Our proof of that is for the most part analogous to Tix's proof of the corresponding result for the three-dimensional operator.    
\end{abstract}

\maketitle

\section{Introduction and partial wave analysis} 
The two-dimensional free Dirac operator is defined in the Hilbert space $L^2\left(\mathbb{R}^2,\mathbb{C}^2\right)$ by
\[
 D = \frac{c \hbar}{i} \sigma_1 \frac{\partial}{\partial x_1} + \frac{c \hbar}{i} \sigma_2 \frac{\partial}{\partial x_2} + m c^2 \sigma_3,
\]
where $\hbar$ denotes Planck's constant, $c$ the speed of light, $m$ the electron mass and 
\[
 \sigma_1 = \begin{pmatrix} 0 & 1 \\ 1 & 0 \end{pmatrix}, \qquad \sigma_2 = \begin{pmatrix} 0 & -i \\ i & 0 \end{pmatrix}, \qquad \sigma_3 = \begin{pmatrix} 1 & 0 \\ 0 & -1 \end{pmatrix}
\] 
the Pauli matrices. $D$ is a self-adjoint operator on $H^1\left(\mathbb{R}^2,\mathbb{C}^2\right)$, therefore we can define
\[
 \Lambda_+ := \chi_{(0,\infty)}(D).
\]
$\Lambda_+$ is obviously a projection and $\mathcal{H}_+ := \Lambda_+\left(L^2\left(\mathbb{R}^2,\mathbb{C}^2\right)\right)$ is called the positive spectral subspace of $D$.

Finally the two-dimensional Brown-Ravenhall operator with Coulomb potential is defined in the Hilbert space $\mathcal{H}_+$ by 
\[
 \mathbf{B} = \Lambda_+\left(D-\frac{\delta}{|\mathbf{x}|}\right)\Lambda_+, \hspace{3ex} \delta>0
\]
By scaling this operator is unitarily equivalent to the one with $m=\hbar=c=1$ up to the factor $mc^2$ and a change of $\delta$ to $\frac{\delta}{\hbar c}$. In the following we will drop $m$, $c$ and $\hbar$. The positivity of the three-dimensional Brown-Ravenhall operator has been proven by Tix \cite{Tix} and, by an alternative method, Burenkov and Evans \cite{BurenkovEvans}. Here we are interested in the two-dimensional result because of its importance in the description of graphene.

Our main theorem is
\begin{theorem}\label{C}
 Let 
\[
 \delta_c := \left(\frac{\Gamma(\tfrac{1}{4})^4}{8\pi^2}+\frac{8\pi^2}{\Gamma(\tfrac{1}{4})^4}\right)^{-1}.
\]
If $\delta\leq\delta_c$, then 
\[
 \mathbf{B} \geq (1-2\delta),
\]
in particular $\mathbf{B}$ is positive. If $\delta > \delta_c$, then $\mathbf{B}$ is unbounded from below.
\end{theorem}
The approximate value of $\delta_c$ is 0.378.

The first part of the proof has been carried out by Bouzouina \cite{Bouzouina}, unfortunately with a small error that leads to a wrong value of $\delta_c$. We proceed to sketch his development without proofs and correct the error where it was introduced.

The free Dirac Operator corresponds in Fourier Space to the $2\times2$ matrix multiplication operator given by
\[
 D(\mathbf{p})=\begin{pmatrix} 1 & p_1-ip_2 \\ p_1+ip_2 & -1 \end{pmatrix},
\]
which has the pointwise eigenvalues $\pm E(\mathbf{p})$, where 
\[
 E(\mathbf{p})=\sqrt{|\mathbf{p}|^2+1}.
\]
$\Lambda_+$ corresponds to the pointwise projection onto the eigenspace of $E(\mathbf{p})$. Thus the Fourier transform of any element of $\mathcal{H}_+$ can be written as $u\cdot\xi$, where $u$ is an element of $L^2\left(\mathbb{R}^2\right)$ and $\xi$ denotes the normed eigenvector associated with $E(\mathbf{p})$. So we can reduce the quadratic form belonging to $\mathbf{B}$ to one over $L^2\left(\mathbb{R}^2\right)$:
\[
 \langle\psi,\mathbf{B}\psi\rangle=\int E(\mathbf{p})|u(\mathbf{p})|^2\;d\mathbf{p} - \frac{\delta}{2\pi}\int\int u(\mathbf{p})\overline{u(\mathbf{p}')}K(\mathbf{p},\mathbf{p}')\;d\mathbf{p}d\mathbf{p}'=:\langle u,\mathbf{b}u\rangle,
\]
where
\[
 K(\mathbf{p},\mathbf{p}') = \frac{(E(\mathbf{p})+1)(E(\mathbf{p}')+1)+\mathbf{p}\overline{\mathbf{p}'}}{N(\mathbf{p})N(\mathbf{p}')|\mathbf{p}-\mathbf{p}'|}
\]
and 
\[
 N(\mathbf{p}):=\sqrt{2E(\mathbf{p})(E(\mathbf{p})+1)}.
\]

For $u\in L^2\left(\mathbb{R}^2\right)$ we can write
\[
 u(\mathbf{p})=u\left(re^{i\theta}\right)=\frac{1}{\sqrt{2\pi}}\sum_{k\in\mathbb{Z}}\frac{1}{\sqrt{r}}a_k(r)e^{ik\theta},
\]
where the $a_k$ are in $L^2((0,\infty))$. Using this decomposition we get
\begin{theorem}\label{A}
 For any $u\in L^2\left(\mathbb{R}^2,\sqrt{1+|\mathbf{p}|^2}\;d\mathbf{p}\right)$ holds
\[
 \langle u,\mathbf{b}u\rangle=\sum_{k\in\mathbb{Z}}\langle a_k,b_k a_k\rangle,
\]
where
\[
 \langle a_k,b_k a_k\rangle = \int_0^\infty e(r)|a_k(r)|^2\;dr-\frac{\delta}{\pi}\int_0^\infty\int_0^\infty a_k(r)\overline{a_k(r')}K_k(r,r')\;drdr'
\]
and for $k\geq 0$
\[
 K_k(r,r')=\beta_1(r,r')Q_{k-1/2}\left(\frac{1}{2}\left(\frac{r}{r'}+\frac{r'}{r}\right)\right)+\beta_2(r,r')Q_{k+1/2}\left(\frac{1}{2}\left(\frac{r}{r'}+\frac{r'}{r}\right)\right),
\]
whereas for $k<0$
\[
 K_k(r,r')=\beta_1(r,r')Q_{-k-1/2}\left(\frac{1}{2}\left(\frac{r}{r'}+\frac{r'}{r}\right)\right)+\beta_2(r,r')Q_{-k-3/2}\left(\frac{1}{2}\left(\frac{r}{r'}+\frac{r'}{r}\right)\right).
\]
We have used
\[
 \beta_1(r,r')=\frac{(e(r)+1)(e(r')+1)}{n(r)n(r')}, \hspace{3ex} \beta_2(r,r')=\frac{rr'}{n(r)n(r')},
\]
and
\[
 e(r)=\sqrt{r^2+1}, \hspace{3ex} n(r)=\sqrt{2e(r)(e(r)+1)}.
\] 
\end{theorem}
The $Q_i$ are Legendre functions of the second kind.
\begin{proof}
 This is Lemma 2.1 and Lemma 2.2 of \cite{Bouzouina} with the difference that the $K_k(r,r')$ of \cite{Bouzouina} are incorrectly twice as big as here. In the proof of Lemma 2.1 Bouzouina claims that for any $q\in]0,1[$ and any $l,l'\in\mathbb{Z}$
\[
 \int_0^{2\pi}\int_0^{2\pi}\frac{e^{il\theta}e^{-il'\theta'}}{\sqrt{1-q\cos(\theta-\theta')}}\;d\theta d\theta' = 4\pi\delta_{l,l'}\int_0^{2\pi}\frac{\cos(l\theta)}{\sqrt{1-q\cos(\theta)}}\;d\theta,
\]
when in fact
\[
 \int_0^{2\pi}\int_0^{2\pi}\frac{e^{il\theta}e^{-il'\theta'}}{\sqrt{1-q\cos(\theta-\theta')}}\;d\theta d\theta' = \int_0^{2\pi}e^{i(l-l')\theta'}\int_0^{2\pi}\frac{e^{il(\theta-\theta')}}{\sqrt{1-q\cos(\theta-\theta')}}\;d\theta d\theta'
\]
\[
 =2\pi\delta_{l,l'}\int_0^{2\pi}\frac{e^{il\theta}}{\sqrt{1-q\cos\theta}}\;d\theta=2\pi\delta_{l,l'}\int_0^{2\pi}\frac{\cos(l\theta)}{\sqrt{1-q\cos\theta}}\;d\theta,
\]
as 
\[
 \int_0^{2\pi}\frac{\sin(l\theta)}{\sqrt{1-q\cos\theta}}\;d\theta=0.
\]
\end{proof}
The Legendre functions $Q_{k-1/2}$, $k\in\mathbb{N}$ occurring in Theorem \ref{A} are positive on $]1,\infty[$ and $(Q_{k-1/2})_{k\in\mathbb{N}}$ is a decreasing sequence (see \cite{Bouzouina}*{Lemma 2.2}). From that it is easy to verify that $K_0\geq K_k$ for all $k\in\mathbb{Z}$. Therefore we have
\begin{corollary}\label{B}
 We have 
\begin{align*}
 \inf\left\{\langle u,\mathbf{b}u\rangle|u\in L^2\left(\mathbb{R}^2,\sqrt{1+|\mathbf{p}|^2}\;d\mathbf{p}\right), ||u||_2=1\right\}\\
=\inf\left\{\langle f,b_0f\rangle|f\in L^2\left((0,\infty),\sqrt{1+r^2}\;dr\right), ||f||_2=1\right\}
\end{align*}
\end{corollary}
\begin{proof}
 This is Corollary 2.4 of \cite{Bouzouina}.
\end{proof}
\section{Positivity of the operator}
In view of Corollary \ref{B} it suffices to prove the equivalent of Theorem \ref{C} for the form $b_0$.
\begin{theorem}
 If $\delta\leq\delta_c$, then 
\begin{equation}\label{D}
 b_0 \geq (1-2\delta).
\end{equation}
\end{theorem}
\begin{proof}
 The first part is an adaptation of the three-dimensional case in \cite{Tix}. Therefore we will be brief. It suffices to prove (\ref{D}) for $\delta=\delta_c$. Then the general case is proven in the following way:
\[
 \langle f,b_0f\rangle=\left(1-\frac{\delta}{\delta_c}\right)\int^\infty_0 e(p)|f(p)|^2\;dp\; 
\]
\[
+\;\frac{\delta}{\delta_c}\left(\int^\infty_0 e(p)|f(p)|^2\;dp - \frac{\delta_c}{\pi}\int^\infty_0 \int^\infty_0 \overline{f(p')}f(p)K_0(p,p')\;dpdp'\right)\geq
\]
\[
\left(1-\frac{\delta}{\delta_c}\right)||f||^2+\frac{\delta}{\delta_c}(1-2\delta_c)||f||^2=(1-2\delta)||f||^2
\]
Finally we need only consider positive $f$. By a simple calculation one gets
\[
 \beta_1(p,p')=\frac{1}{2}\sqrt{1+\frac{1}{e(p)}}\sqrt{1+\frac{1}{e(p')}}
\]
and
\[
\beta_2(p,p')=\frac{1}{2}\sqrt{1-\frac{1}{e(p)}}\sqrt{1-\frac{1}{e(p')}}
\]
By the usual application of the Cauchy-Schwarz inequality introducing positive trial functions $h_0$ and $h_1$ (for details see \cite{Tix}) we have
\[
 \langle f,b_0f\rangle \geq \underset{p\in[0,\infty)}{\text{inf}}\mathcal{E}(p)\cdot ||f||^2, 
\]
where
\[
 \mathcal{E}(p):= e(p)-\frac{\delta_c}{2\pi}\left(\left(1+\frac{1}{e(p)}\right)\frac{I_0(p)}{h_0(p)}+\left(1-\frac{1}{e(p)}\right)\frac{I_1(p)}{h_1(p)}\right)
\]
and
\[
 I_k(p):=\int^\infty_0 h_k(p')Q_{k-1/2}\left(\frac{1}{2}\left(\frac{p}{p'}+\frac{p'}{p}\right)\right)dp'.
\]
It remains to choose $h_0$ and $h_1$ such that this infimum is $1-2\delta_c$. Note that choosing $h_0(p)=h_1(p)=\frac{1}{p}$ as Evans et al. did in \cite{Eps} would result in the infimum being finite, i.e. boundedness of the operator from below, but negative. 
Similarly to Tix we look for $k=0,1$ at the functions $g_k(re^{i\theta}):=r^{a}e^{-r}e^{ik\theta}$with $a>-2$, so that the $g_k$ are in $L^1\left(\mathbb{R}^2\right)$. Let the Fourier transform on $\mathbb{R}^2$ be defined as usual by
\[
 \widehat{f}(\mathbf{p}):=\frac{1}{2\pi}\int_{\mathbb{R}^2}e^{-i\mathbf{p}\cdot\mathbf{x}}f(\mathbf{x})\;d\mathbf{x}.
\]
It is a well known theorem \cite{Steinweiss}*{Theorem IV.1.6} that for $f\in L^1\left(\mathbb{R}^2\right)$ of the form $f(re^{i\theta})=f_0(r)e^{ik\theta}$ the Fourier transform is again of this form, more precisely:
\[
 \widehat{f}(pe^{ik\phi})=e^{ik\phi}(-i)^k\int_0^\infty J_{k}(pr)f_0(r)r\;dr,
\]
where the $J_k$ are Bessel functions of the first order. With \cite{Gradshteyn}*{6.621.1} it follows that
\begin{equation}\label{E}
 \widehat{g_k}(pe^{i\phi})=e^{ik\phi}(-i)^k(p^2+1)^{-\frac{a+2}{2}}\Gamma(k+a+2)P^{-k}_{a+1}\left((p^2+1)^{-\frac{1}{2}}\right).
\end{equation}
Let now $a=-\frac{1}{2}$.
From the integral representation \cite{Gradshteyn}*{8.714.1}
\[
 P_\nu^\mu(\cos\phi)=\sqrt{\frac{2}{\pi}}\frac{\sin^\mu(\phi)}{\Gamma\left(\frac{1}{2}-\mu\right)}\int_0^\phi\frac{\cos\left(\left(\nu+\frac{1}{2}\right)t\right)\;dt}{(\cos t-\cos\phi)^{\mu+\frac{1}{2}}}, \hspace{2ex} 0<\phi<\pi,\hspace{1ex} \mu<\frac{1}{2}
\]
and the fact that the gamma function is positive for positive argument it follows that the radial part of $g_k$, i.e. 
\[
 f_k(p):=(p^2+1)^{-\frac{3}{4}}\Gamma\left(k+\frac{3}{2}\right)P^{-k}_{-1/2}\left((p^2+1)^{\frac{1}{2}}\right)
\]
is positive. Now we calculate the Fourier transform of $\frac{1}{|x|}g_k$ in two different ways. On the one hand
\[
 \left(\frac{1}{|\cdot|}\cdot g_k\right)^{\widehat{}}(pe^{i\phi}) = \frac{1}{2\pi}\left(\widehat{\frac{1}{|\cdot|}}*\widehat{g_k}\right)(pe^{i\phi}) =  \frac{1}{2\pi}\int_{\mathbb{R}^2}\frac{\widehat{g_k}(\mathbf{p}')}{|pe^{i\phi}-\mathbf{p}'|}\,d\mathbf{p'} = 
\]
\[
=\frac{1}{2\pi}e^{ik\phi}(-i)^k\int_0^\infty\int_0^{2\pi}p'\frac{e^{ik\phi'}f_k(p')}{\sqrt{p^2+p'^2-2pp'cos\phi'}}\;d\phi'dp'
\]
\[
= \frac{1}{\pi \sqrt{p}}e^{ik\phi}(-i)^k\int_0^\infty \sqrt{p'}f_k(p')Q_{k-1/2}\left(\frac{1}{2}\left(\frac{p}{p'}+\frac{p'}{p}\right)\right)\;dp', 
\]
where the last identity follows from Lemma 2.2 of \cite{Bouzouina}. On the other hand, setting $a=-\frac{3}{2}$ instead of $a=-\frac{1}{2}$ in (\ref{E}), we get
\[
 \left(\frac{1}{|\cdot|}\cdot g_k\right)^{\widehat{}}(pe^{i\phi}) = e^{ik\phi}(-i)^k(p^2+1)^{-\frac{1}{4}}\Gamma\left(k+\frac{1}{2}\right)P^{-k}_{-1/2}\left((p^2+1)^{-\frac{1}{2}}\right).
\]
Now let
\[
 h_k(p):=\sqrt{p}(p^2+1)^{-\frac{3}{4}}\Gamma\left(k+\frac{3}{2}\right)P^{-k}_{1/2}\left((p^2+1)^{-\frac{1}{2}}\right),
\]
which is positive. Then 
\[
 I_k(p) = \pi\sqrt{p}(p^2+1)^{-\frac{1}{4}}\Gamma\left(k+\frac{1}{2}\right)P^{-k}_{-1/2}\left((p^2+1)^{-\frac{1}{2}}\right).
\]
Therefore
\[
 \mathcal{E}(p)=e(p)-\delta_c\left((e(p)+1)\frac{P_{-1/2}\left(\frac{1}{e(p)}\right)}{P_{1/2}\left(\frac{1}{e(p)}\right)}+(e(p)-1)\cdot\frac{1}{3}\cdot\frac{P^{-1}_{-1/2}\left(\frac{1}{e(p)}\right)}{P^{-1}_{1/2}\left(\frac{1}{e(p)}\right)}\right)
\]
and 
\[
 \inf_{p\in[0,\infty)}\mathcal{E}(p) = \inf_{x\in(0,1]}f(x),
\]
where
\[
 f(x):=\left( \frac{1}{x}-\delta_c\left(\left(\frac{1}{x}+1\right)\frac{P_{-1/2}(x)}{P_{1/2}(x)}+\left(\frac{1}{x}-1\right)\cdot\frac{1}{3}\cdot\frac{P^{-1}_{-1/2}(x)}{P^{-1}_{1/2}(x)}\right)\right).
\]

Now we proceed to show that $\inf_{x\in(0,1]}f(x) = 1-2\delta_c$. According to \cite{Gradshteyn}*{8.704} we have 
\[
 P^\mu_\nu(x)=\frac{1}{\Gamma(1-\mu)}\left(\frac{1+x}{1-x}\right)^{\frac{\mu}{2}}F\left(-\nu,\nu+1;1-\mu;\frac{1-x}{2}\right) 
\]
and therefore
\[
 f(x)=\left( \frac{1}{x}-\delta_c\left(\left(\frac{1}{x}+1\right)\frac{F\left(\frac{1}{2},\frac{1}{2};1;\frac{1-x}{2}\right)}{F\left(-\frac{1}{2},\frac{3}{2};1;\frac{1-x}{2}\right)}+\left(\frac{1}{x}-1\right)\cdot\frac{1}{3}\cdot\frac{F\left(\frac{1}{2},\frac{1}{2};2;\frac{1-x}{2}\right)}{F\left(-\frac{1}{2},\frac{3}{2};2;\frac{1-x}{2}\right)}\right)\right).
\]
By definition
\[
 F(a,b;c;x)=\sum_{k=0}^{\infty}\frac{(a)_k(b)_k}{(c)_kk!}x^k
\]
for $|x|<1$, where
\[
 (a)_k:=a(a+1)\cdot\cdot\cdot(a+n-1), \hspace{3ex} (a)_0:=1.
\]
So we have for $0<x\leq 1$
\[
 F\left(-\nu,\nu+1;1-\mu;\frac{1-x}{2}\right)=\sum_{k=0}^{\infty}\frac{(-\nu)_k(\nu+1)_k}{(1-\mu)_kk!2^k}(1-x)^k.
\]
We see at once, that $f(1)=1-2\delta_c$. Examining the series above for $\nu=\pm\frac{1}{2}$ and $\mu=0,-1$ we notice that for $\nu=-\frac{1}{2}$ all terms are positive, whereas for $\nu=\frac{1}{2}$ all terms except the first one are negative. Furthermore we can estimate
\[
 \left|\sum_{k=3}^{\infty}\frac{(-\nu)_k(\nu+1)_k}{(1-\mu)_kk!2^k}(1-x)^k\right|\leq\left|\frac{(-\nu)_3(\nu+1)_3}{(1-\mu)_33!2^3}\right|(1-x)^3\cdot2\leq \left|\frac{(-\nu)_2(\nu+1)_2}{(1-\mu)_22!2^2}\right|(1-x)^2.
\]
So we have
\[
 f(x)-(1-2\delta_c)\geq\frac{1}{x}-\delta_c\left(\frac{1}{x}+1\right)\frac{1+\frac{1}{8}(1-x)+\frac{9}{128}(1-x)^2}{1-\frac{3}{8}(1-x)-\frac{15}{128}(1-x)^2}-
\]
\[
-\left(\frac{1}{x}-1\right)\cdot\frac{1}{3}\cdot\frac{1+\frac{1}{16}(1-x)+\frac{3}{128}(1-x)^2}{1-\frac{3}{16}(1-x)-\frac{5}{128}(1-x)^2}-(1-2\delta_c).
\]
After some calculation this equals
\[
 ((49152-114688\delta_c)(1-x)+(-27648+48128\delta_c)(1-x)^2+(-4224+16128\delta_c)(1-x)^3+
\]
\[
+(1800-3504\delta_c)(1-x)^4+(225-540\delta_c)(1-x)^5)/
\]
\[
(3x(128-48(1-x)-15(1-x)^2)(128-24(1-x)-5(1-x)^2)).
\]
The denominator is obviously positive and so are $-4224+16128\delta_c$, $1800-3504\delta_c$ and $225-540\delta_c$. Thus the term above is positive, if 
\[
 (49152-114688\delta_c)+(-27648+48128\delta_c)(1-x)\geq0
\]
and this is the case if $x\geq0.4$. Therefore $f(x)\geq(1-2\delta_c)$ for $x\geq0.4$.

To show that $f(x)\geq(1-2\delta_c)$ for $x\leq0.4$ we need the series expansion at $0$ \cite{Gradshteyn}*{8.775}:
\[
  P^\mu_\nu(x)=\frac{2^\mu \cos\left(\frac{1}{2}(\nu+\mu)\pi\right)\Gamma\left(\frac{\nu+\mu+1}{2}\right)}{\sqrt{\pi}\Gamma\left(\frac{\nu-\mu}{2}+1\right)}(1-x^2)^{\frac{\mu}{2}}\sum_{k=0}^\infty \frac{\left(\frac{\nu+\mu+1}{2}\right)_k\left(\frac{\nu-\mu}{2}\right)_k}{\left(\frac{1}{2}\right)_k k!}x^{2k}+
\]
\[
+\frac{2^{\mu+1} \sin\left(\frac{1}{2}(\nu+\mu)\pi\right)\Gamma\left(\frac{\nu+\mu}{2}+1\right)}{\sqrt{\pi}\Gamma\left(\frac{\nu-\mu+1}{2}\right)}x(1-x^2)^{\frac{\mu}{2}}\sum_{k=0}^\infty \frac{\left(\frac{\nu+\mu}{2}+1\right)_k\left(\frac{-\nu+\mu+1}{2}\right)_k}{\left(\frac{3}{2}\right)_k k!}x^{2k}
\]
for $0\leq x<1$. So we have
\[
 P_{-1/2}(x)=\frac{1}{2\pi}\sum_{k=0}^\infty\left( \frac{\Gamma\left(\frac{1}{4}+k\right)^2}{\Gamma\left(\frac{1}{2}+k\right)k!}x^{2k}-\frac{\Gamma\left(\frac{3}{4}+k\right)^2}{\Gamma\left(\frac{3}{2}+k\right)k!}x^{2k+1}\right)
\]
\[
 P_{1/2}(x)=\frac{1}{2\pi}\sum_{k=0}^\infty\left(-\frac{\Gamma\left(\frac{3}{4}+k\right)\Gamma\left(-\frac{1}{4}+k\right)}{\Gamma\left(\frac{1}{2}+k\right)k!}x^{2k}+\frac{\Gamma\left(\frac{1}{4}+k\right)\Gamma\left(\frac{5}{4}+k\right)}{\Gamma\left(\frac{3}{2}+k\right)k!}x^{2k+1}\right)
\]
\[
 \sqrt{1-x^2}P_{-1/2}^{-1}(x)=\frac{1}{4\pi}\sum_{k=0}^\infty\left(\frac{\Gamma\left(-\frac{1}{4}+k\right)^2}{\Gamma\left(\frac{1}{2}+k\right)k!}x^{2k}-\frac{\Gamma\left(\frac{1}{4}+k\right)^2}{\Gamma\left(\frac{3}{2}+k\right)k!}x^{2k+1}\right)
\]
\[
 \sqrt{1-x^2}P_{1/2}^{-1}(x)=\frac{1}{4\pi}\sum_{k=0}^\infty\left(-\frac{\Gamma\left(\frac{1}{4}+k\right)\Gamma\left(-\frac{3}{4}+k\right)}{\Gamma\left(\frac{1}{2}+k\right)k!}x^{2k}+\frac{\Gamma\left(\frac{3}{4}+k\right)\Gamma\left(-\frac{1}{4}+k\right)}{\Gamma\left(\frac{3}{2}+k\right)k!}x^{2k+1}\right)
\]
To arrive at an estimate we first show that the function mapping $x$ to $\frac{\Gamma(x)}{\Gamma\left(x+\frac{1}{2}\right)}$ is strictly decreasing on $(0,\infty)$. For $x>0$ 
\[
 \frac{d}{dx}\frac{\Gamma(x)}{\Gamma\left(x+\frac{1}{2}\right)}=\frac{\Gamma'(x)\Gamma\left(x+\frac{1}{2}\right)-\Gamma(x)\Gamma'\left(x+\frac{1}{2}\right)}{\Gamma\left(x+\frac{1}{2}\right)^2}=\frac{\Gamma(x)\psi(x)-\Gamma(x)\psi\left(x+\frac{1}{2}\right)}{\Gamma\left(x+\frac{1}{2}\right)},
\]
where 
\[
 \psi(x):=\frac{\Gamma'(x)}{\Gamma(x)}.
\]
According to \cite{Gradshteyn}*{8.361}
\[
 \psi(x)=\int_0^\infty\left(\frac{e^{-t}}{t}-\frac{e^{-xt}}{1-e^{-t}}\right)dt
\]
and therefore $\psi$ is strictly increasing and $\frac{\Gamma(x)}{\Gamma\left(x+\frac{1}{2}\right)}$ is strictly decreasing on $(0,\infty)$. If we denote the nth coefficient of the respective power series by $a_n$, we get in each of the four cases for $k\geq 1$
\[
 \left|\frac{a_{2k}}{a_{2k+1}}\right|\geq\frac{\Gamma\left(\frac{3}{4}+k\right)\Gamma\left(\frac{1}{4}+k\right)\Gamma\left(\frac{3}{2}+k\right)}{\Gamma\left(\frac{5}{4}+k\right)\Gamma\left(\frac{3}{4}+k\right)\Gamma\left(\frac{1}{2}+k\right)}=\frac{k+\frac{1}{2}}{k+\frac{1}{4}}>1
\]
and
\[
 \left|\frac{a_{2k+1}}{a_{2k+2}}\right|\geq\frac{\Gamma\left(\frac{5}{4}+k\right)\Gamma\left(\frac{3}{4}+k\right)(k+1)!}{\Gamma\left(\frac{7}{4}+k\right)\Gamma\left(\frac{5}{4}+k\right)k!}=\frac{k+1}{k+\frac{3}{4}}>1.
\]
So starting from $a_2$, the coefficients alternate and their absolute values decrease strictly. Therefore
\[
 P_{-1/2}(x)\leq \frac{1}{2\pi}\left(\frac{\Gamma\left(\frac{1}{4}\right)^2}{\sqrt{\pi}}-\frac{4\pi^{\frac{3}{2}}}{\Gamma\left(\frac{1}{4}\right)^2}x+\frac{\Gamma\left(\frac{1}{4}\right)^2}{8\sqrt{\pi}}x^2\right)
\]
\[
 P_{1/2}(x)\geq \frac{1}{2\pi}\left(\frac{8\pi^{\frac{3}{2}}}{\Gamma\left(\frac{1}{4}\right)^2}+\frac{\Gamma\left(\frac{1}{4}\right)^2}{2\sqrt{\pi}}x-\frac{3\pi^{\frac{3}{2}}}{\Gamma\left(\frac{1}{4}\right)^2}x^2\right)
\]
\[
 \sqrt{1-x^2}P_{-1/2}^{-1}(x)\leq \frac{1}{4\pi}\left(\frac{32\pi^{\frac{3}{2}}}{\Gamma\left(\frac{1}{4}\right)^2}-\frac{2\Gamma\left(\frac{1}{4}\right)^2}{\sqrt{\pi}}x+\frac{4\pi^{\frac{3}{2}}}{\Gamma\left(\frac{1}{4}\right)^2}x^2\right)
\] 
\[
 \sqrt{1-x^2}P_{1/2}^{-1}(x)\geq \frac{1}{4\pi}\left(\frac{4\Gamma\left(\frac{1}{4}\right)^2}{3\sqrt{\pi}}-\frac{16\pi^{\frac{3}{2}}}{\Gamma\left(\frac{1}{4}\right)^2}x-\frac{\Gamma\left(\frac{1}{4}\right)^2}{2\sqrt{\pi}}x^2\right).
\]
One can easily check that each of the right sides is positive for $x\leq0.4$. Plugging this in and cancelling we get
\[
 f(x)-(1-2\delta_c)\geq\frac{1}{x}-\left(\frac{\Gamma(\tfrac{1}{4})^4}{8\pi^2}+\frac{8\pi^2}{\Gamma(\tfrac{1}{4})^4}\right)^{-1}\left(\left(\frac{1}{x}+1\right)\frac{\Gamma\left(\frac{1}{4}\right)^4-4\pi^2x+\tfrac{1}{8}\Gamma\left(\frac{1}{4}\right)^4x^2}{8\pi^2+\tfrac{1}{2}\Gamma\left(\frac{1}{4}\right)^4x-3\pi^2x^2}-\right.
\]
\[
\left.-\left(\frac{1}{x}-1\right)\frac{8\pi^2-\tfrac{1}{2}\Gamma\left(\frac{1}{4}\right)^4x+\pi^2x^2}{\Gamma\left(\frac{1}{4}\right)^4-12\pi^2x-\tfrac{3}{8}\Gamma\left(\frac{1}{4}\right)^4x^2}\right)-\left(1-2\left(\frac{\Gamma(\tfrac{1}{4})^4}{8\pi^2}+\frac{8\pi^2}{\Gamma(\tfrac{1}{4})^4}\right)^{-1}\right).
\]
The right side equals
\[
\frac{\sum_{n=0}^4b_nx^n}{\left(\Gamma\left(\frac{1}{4}\right)^8+64\pi^4\right)\left(8\Gamma\left(\frac{1}{4}\right)^4-96\pi^2x-3\Gamma\left(\frac{1}{4}\right)^4x^2\right)\left(16\pi^2+\Gamma\left(\frac{1}{4}\right)^4x-6\pi^2x^2\right)},
\]
where
\[
 b_0 = 8\; \Gamma\left(\frac{1}{4}\right)^{16} - 256\; \Gamma\left(\frac{1}{4}\right)^{12} \pi^2 + 3072\; \Gamma\left(\frac{1}{4}\right)^8 \pi^4 - 98304 \pi^8
\]
\[
 b_1 = -8\; \Gamma\left(\frac{1}{4}\right)^{16} + 3072 \;\Gamma\left(\frac{1}{4}\right)^8 \pi^4 - 40960 \;\Gamma\left(\frac{1}{4}\right)^4 \pi^6 + 98304 \pi^8
\]
\[
 b_2 = -3 \; \Gamma\left(\frac{1}{4}\right)^{16} + 192\; \Gamma\left(\frac{1}{4}\right)^{12} \pi^2 - 2944 \; \Gamma\left(\frac{1}{4}\right)^8 \pi^4 + 4096 \;\Gamma\left(\frac{1}{4}\right)^4 \pi^6 + 
 36864 \pi^8
\]
\[
 b_3 = 3\; \Gamma\left(\frac{1}{4}\right)^{16} - 24\; \Gamma\left(\frac{1}{4}\right)^{12} \pi^2 - 128\; \Gamma\left(\frac{1}{4}\right)^8 \pi^4 + 10752\;\Gamma\left(\frac{1}{4}\right)^4 \pi^6 - 
 36864 \pi^8
\]
\[
 b_4 = -12\; \Gamma\left(\frac{1}{4}\right)^{12} \pi^2 + 288\; \Gamma\left(\frac{1}{4}\right)^8 \pi^4 - 1536\;\Gamma\left(\frac{1}{4}\right)^4 \pi^6
\]
One easily verifies that in the fraction above the denominator is positive as well as $b_0$ and $b_3$, while $b_1$, $b_2$ and $b_4$ are negative, so
\[
 \sum_{n=0}^4b_nx^n\geq b_0-b_1\cdot 0.4 -b_2\cdot 0.4^2-b_4\cdot 0.4^4>0.
\]   
\end{proof}
\section{Unboundedness above the critical constant}
The following is inspired by \cite{Eps}.
\begin{lemma}\label{F}
 Let $\nu=\pm\frac{1}{2}$ and
\[
 I_{\nu}=\int^\infty_0\frac{1}{p}\;Q_{\nu}\left(\frac{1}{2}\left(p+\frac{1}{p}\right)\right)dp.
\]
There exists $C\in\mathbb{R}$ such that for all $1<a<b$
\[
 \int^b_a\int^b_a\frac{1}{p}\;\frac{1}{p'}\;Q_{\nu}\left(\frac{1}{2}\left(\frac{p}{p'}+\frac{p'}{p}\right)\right)\;dp'dp\geq I_\nu \cdot\log\left(\frac{a}{b}\right)+C
\]
\end{lemma}
\begin{proof}
 According to \cite{Whittaker}*{p.334} for $t>1$ and $\nu>-1$ we have the integral representation
\[
 Q_{\nu}(t)=\int^{t-\sqrt{t^2-1}}_0\frac{x^{\nu}}{\sqrt{x^2-2tx+1}}\;dx\leq (t-\sqrt{t^2-1}) \int^{t-\sqrt{t^2-1}}_0\frac{x^{\nu-1}}{\sqrt{x^2-2tx+1}}\;dx.
\]
If $0<p<1$ and $t=\frac{1}{2}\left(p+\frac{1}{p}\right)$, then $t-\sqrt{t^2-1}=p$; we conclude that
\[
 Q_{\nu+1}\left(\frac{1}{2}\left(p+\frac{1}{p}\right)\right)\leq p\;Q_{\nu}\left(\frac{1}{2}\left(p+\frac{1}{p}\right)\right).
\]
According to \cite{Gradshteyn}*{8.832.3} the derivative of the Legendre functions is given on $(1,\infty)$ by
\[
 \frac{d}{dx}Q_{\nu}(x) = \frac{(\nu+1)(Q_{\nu+1}(x)-x\;Q_\nu(x))}{x^2-1}.
\]
Let $\alpha\in\mathbb{R}$ and $0<p<1$:
\[
\frac{d}{dp}\left\{Q_{\nu}\left(\frac{1}{2}\left(p+\frac{1}{p}\right)\right)\cdot p^\alpha\right\}
\]
\[
=\frac{(\nu+1)\left(Q_{\nu+1}\left(\frac{1}{2}\left(p+\frac{1}{p}\right)\right)-\frac{1}{2}\left(p+\frac{1}{p}\right)Q_{\nu}\left(\frac{1}{2}\left(p+\frac{1}{p}\right)\right)\right)\cdot \frac{1}{2}\left(1-\frac{1}{p^2}\right)}{\frac{1}{4}\left(p^2+\frac{1}{p^2}\right)-\frac{1}{2}}p^\alpha
\]
\[
+\alpha\; Q_{\nu}\left(\frac{1}{2}\left(p+\frac{1}{p}\right)\right)p^{\alpha-1}
\]
\[
=\frac{2(\nu+1)\left(\frac{1}{2}\left(p+\frac{1}{p}\right)Q_{\nu}\left(\frac{1}{2}\left(p+\frac{1}{p}\right)\right)-Q_{\nu+1}\left(\frac{1}{2}\left(p+\frac{1}{p}\right)\right)\right)\cdot\left(\frac{1}{p^2}-1\right)}{\left(\frac{1}{p}-p\right)^2}p^\alpha 
\]
\[
+ \alpha\; Q_{\nu}\left(\frac{1}{2}\left(p+\frac{1}{p}\right)\right)p^{\alpha-1}
\]
\[
\geq\frac{2(\nu+1)\left(\frac{1}{2}\left(p+\frac{1}{p}\right)Q_{\nu}\left(\frac{1}{2}\left(p+\frac{1}{p}\right)\right)-p\;Q_{\nu}\left(\frac{1}{2}\left(p+\frac{1}{p}\right)\right)\right)}{\frac{1}{p}-p}p^{\alpha-1}
\]
\[
+ \alpha\; Q_{\nu}\left(\frac{1}{2}\left(p+\frac{1}{p}\right)\right)p^{\alpha-1}
\]
\[
=(\nu+1+\alpha)\;Q_{\nu}\left(\frac{1}{2}\left(p+\frac{1}{p}\right)\right)p^{\alpha-1}.
\]
Therefore the function
\[
 p\mapsto Q_{\nu}\left(\frac{1}{2}\left(p+\frac{1}{p}\right)\right)\cdot p^\alpha
\]
is increasing on $(0,1)$, if $\alpha\geq-\nu-1$, which is satisfied by $\nu=\pm\frac{1}{2}$ and $\alpha=-\frac{1}{2}$. Thus
\[
\int^b_a\frac{1}{p}\;\int^a_0\frac{1}{p'}\;Q_{\nu}\left(\frac{1}{2}\left(\frac{p}{p'}+\frac{p'}{p}\right)\right)\;dp'dp=\int^b_a\frac{1}{p}\;\int^{\tfrac{a}{p}}_0\frac{1}{p'}\;Q_{\nu}\left(\frac{1}{2}\left(p'+\frac{1}{p'}\right)\right)\;dp'dp
\]
\[
\leq\int^b_a\frac{1}{p}\;\sqrt{\frac{p}{a}}\;Q_{\nu}\left(\frac{1}{2}\left(\frac{a}{p}+\frac{p}{a}\right)\right)\int^{\tfrac{a}{p}}_0\frac{1}{\sqrt{p'}}dp'dp=2\int^{\tfrac{b}{a}}_1\frac{1}{p}\;Q_{\nu}\left(\frac{1}{2}\left(p+\frac{1}{p}\right)\right)dp\leq I_\nu
\]
and
\[
 \int^b_a\frac{1}{p}\;\int^\infty_b\frac{1}{p'}\;Q_{\nu}\left(\frac{1}{2}\left(\frac{p}{p'}+\frac{p'}{p}\right)\right)\;dp'dp=\int^b_a\frac{1}{p}\;\int^{\tfrac{p}{b}}_0\frac{1}{p'}\;Q_{\nu}\left(\frac{1}{2}\left(p'+\frac{1}{p'}\right)\right)\;dp'dp
\]
\[
\leq2\int^b_a\frac{1}{p}\;Q_{\nu}\left(\frac{1}{2}\left(\frac{p}{b}+\frac{b}{p}\right)\right)\;dp=2\int^{\tfrac{b}{a}}_1\frac{1}{p}\;Q_{\nu}\left(\frac{1}{2}\left(p+\frac{1}{p}\right)\right)dp\leq I_\nu.
\]
From this the assertion of the lemma follows with $C=-2I_\nu$.
\end{proof}
\begin{theorem}
 If $\delta>\delta_c$, $b_0$ is unbounded from below.
\end{theorem}
\begin{proof}
 Let $1<a<b$ and define for $p>0$
\[
 f(p):=\chi_{(a,b)}(p)\cdot \frac{1}{p}.
\]
We have for $p,p'\in(a,b)$
\[
 \beta_1(p,p')\geq\frac{1}{2}
\]
and
\[
 \beta_2(p,p')=\frac{1}{2}\sqrt{1-\frac{1}{e(p)}}\sqrt{1-\frac{1}{e(p')}} \geq \frac{1}{2}\left(1-\frac{1}{e(p)}\right)\left(1-\frac{1}{e(p')}\right)
\]
\[
 \geq\frac{1}{2}\left(1-\frac{1}{p}-\frac{1}{p'}\right)\geq\frac{1}{2}-\frac{1}{a}.
\]
From \cite{Bouzouina} we know that
\[
 I_{-1/2}=\frac{\Gamma\left(\frac{1}{4}\right)^4}{4\pi}\hspace{2ex} \text{and} \hspace{2ex} I_{1/2}=\frac{16\pi^3}{\Gamma\left(\frac{1}{4}\right)^4}.
\]
The facts collected above and Lemma \ref{F} are used in the following estimate. From now on we replace some terms that do not depend on $a$ or $b$ by the word 'const':
\[
\frac{\delta}{\pi}\int^b_a\int^b_a\frac{1}{p}\;\frac{1}{p'}\;K_0(p,p')\;dpdp'
\]
\[
\geq\frac{\delta}{2\pi}\left(\frac{\Gamma\left(\tfrac{1}{4}\right)^4}{4\pi}\;\ln\left(\frac{b}{a}\right)+\frac{16\pi^3}{\Gamma\left(\tfrac{1}{4}\right)^4}\;\ln\left(\frac{b}{a}\right)\left(1-\frac{2}{a}\right)\right)+\text{const}
\]
\[
\geq\delta\left(\frac{\Gamma\left(\tfrac{1}{4}\right)^4}{8\pi^2}+\frac{8\pi^2}{\Gamma\left(\tfrac{1}{4}\right)^4}\right)\ln\left(\frac{b}{a}\right)\left(1-\frac{2}{a}\right)+\text{const}
\]
\[
=\frac{\delta}{\delta_c}\ln\left(\frac{b}{a}\right)\left(1-\frac{2}{a}\right)+\text{const}.  
\]
Furthermore
\[
 \int^b_a e(p)\frac{1}{p^2}\;dp\leq \int^b_a \frac{p+1}{p^2}\;dp\leq \ln\left(\frac{b}{a}\right)+1.
\]
Thus 
\[
 \langle f,b_0f\rangle \leq c\left(1-\frac{\delta}{\delta_c}\left(1-\frac{2}{a}\right)\right)\ln\left(\frac{b}{a}\right)+\text{const}.
\]
Since $\frac{\delta}{\delta_c}>1$, we can choose $a$ large enough so that
\[
 \left(1-\frac{\delta}{\delta_c}\left(1-\frac{2}{a}\right)\right)<0.
\]
If we keep $a$ fixed and let $b$ go to infinity, $\langle f,b_0f\rangle$ tends to negative infinity, while $||f||^2=\left(\tfrac{1}{a}- \tfrac{1}{b}\right)\leq 1$. This concludes the proof.
\end{proof}
The proof of Theorem \ref{C} is now complete.
\subsection*{Acknowledgement}
The author thanks H. Siedentop for giving valuable advice.
  
\bibliographystyle{amsplain}

\end{document}